\documentclass[letterpaper, 10 pt, conference]{ieeeconf}  % Comment this line out if you need a4paper

\IEEEoverridecommandlockouts                              % This command is only needed if 
\usepackage{graphics} % for pdf, bitmapped graphics files
\usepackage{epsfig} % for postscript graphics files
\usepackage{mathptmx} % assumes new font selection scheme installed
\usepackage{times} % assumes new font selection scheme installed

\usepackage{amsmath,amssymb,bm,bbm,mathrsfs,amscd}
\usepackage{color}
\usepackage{dsfont}
\usepackage{graphicx}
\usepackage{tikz}

\usepackage{caption}
\def\beq{\begin{equation}}
\def\eeq{\end{equation}}
\newtheorem{theorem}{Theorem}

\newtheorem{proposition}[theorem]{Proposition}
\newtheorem{lemma}[theorem]{Lemma}
\newtheorem{corollary}[theorem]{Corollary}
\newtheorem{example}{Example}%}

\newcommand{\ba}{\begin{array}}
\newcommand{\ea}{\end{array}}

\newcommand{\be}{\begin{equation}}
\newcommand{\ee}{\end{equation}}

\newcommand{\mc}{\mathcal}

\newcommand{\1}{\mathbbm{1}}

\newcommand{\N}{\mathbb{N}}

\renewcommand{\span}{\text{span}}

\DeclareMathOperator*{\argmax}{argmax}

\def\1{\mathds{1}}

\def\N{\mathbb{N}}

\title{\LARGE \bf A game theoretic approach to a peer-to-peer cloud storage model}

\author{Fabio Fagnani$^{1}$,  Barbara Franci$^{2}$, and Ennio Grasso$^{3}$% <-this % stops a space
\thanks{$^{1}$Fabio Fagnani is with the Department of Mathematical Sciences, Politecnico di Torino, 
        {\tt\small fabio.fagnani@polito.it}}%
\thanks{$^{2}$Barbara Franci is with the Department of Mathematical Sciences, Politecnico di Torino,
        {\tt\small barbara.franci@polito.it}}%
\thanks{$^{3}$Ennio Grasso is with Swarm Lab - Joint Open Lab TelecomItalia, Torino
        {\tt\small ennio.grasso@telecomitalia.it }}%
}

\begin{document}

\maketitle
\thispagestyle{empty}
\pagestyle{empty}

%%%%%%%%%%%%%%%%%%%%%%%%%%%%%%%%%%%%%%%%%%%%%%%%%%%%%%%%%%%%%%%%%%%%%%%%%%%%%%%%
\begin{abstract}
Classical cloud storage based on external data providers has been recognized to suffer from a number of drawbacks. This is due to its inherent centralized architecture which makes it vulnerable to external attacks, malware, technical failures, as well to the large premium charged for business purposes.  In this paper, we propose an alternative distributed peer-to-peer cloud storage model which is based on the observation that the users themselves often have available storage capabilities to be offered in principle to other users. Our set-up is that of a network of users connected through a graph, each of them being at the same time a source of data to be stored externally and a possible storage resource. We cast the peer-to-peer storage model to a Potential Game and we propose an original decentralized algorithm which makes units interact, cooperate, and store a complete back up of their data on their connected neighbors. We present theoretical results on the algorithm as well a good number of simulations which validate our approach.

\end{abstract}

%%%%%%%%%%%%%%%%%%%%%%%%%%%%%%%%%%%%%%%%%%%%%%%%%%%%%%%%%%%%%%%%%%%%%%%%%%%%%%%%
\section{INTRODUCTION}

% ENNIO's INTRODUCTION

Cloud storage on the Internet has come to rely almost exclusively on data providers serving as trusted third parties to transfer and store the data. While the system works well enough in most cases, it still suffers from the inherent weaknesses of the trust-based model. The traditional cloud is open to a variety of security threats, including man-in-the-middle attacks, and malware that expose sensitive and private consumer and corporate data. Furthermore, current cloud storage applications are charging large premiums on data storage facilities to business clients. Moreover, these cloud storage providers may have technical failures that can cause data breaches and unavailability, much to the distress of the users and applications that depend on them.

To address the aforementioned shortcomings, a decentralized peer-to-peer cloud storage model would be the right answer. On the wake of the successful peer-to-peer file sharing model of applications like BitTorrent and its lookalike, the same philosophy may well be leveraged on a different but very similar application service like storage.
Indeed a slew of fledgling and somehow successful startups are entering in this market niche. Among the most noteworthy examples are:
\begin{itemize}
\item Storj: www.youtube.com/channel/UC-cTEqWwZV5Rl-h0RZsp2Qw,
\item BitTorrent Sync:  www.getsync.com/,
\item Ethos:  www.youtube.com/watch?v=qUftGCQ5dqo,
\item SpaceMonkeys: www.spacemonkey.com/.
\end{itemize}
Clearly, a completely decentralized peer-to-peer model must account for some challenging technical difficulties that are absent in a centralized cloud model. Firstly, security and privacy must be carefully implemented by ensuring end-to-end encryption resistant to attackers. In addition, the model must account for the latency, performance, and downtime of average user devices.

Albeit the above technical issues are challenging, they can be addressed with the right tools and architectures available at current state of the art technology and will not be considered in this paper. What remains an open question up to now is how to endow the system with the right incentives for the end users to collaborate and share their storage commodity with each other. We believe that the answer to that question comes from the formal framework of game-theory which provides the mathematical tools to ensure successful cooperation among users / players. There are (at least) two ways in which game-theory, with its plenty folk-theorems, can be applied to the real world. One is as a tool to study an ongoing phenomenon. This is the typical setting of social and psychological sciences. A second more engineering approach is to leverage game-theory to design specific mechanisms, i.e. set rules of the game that will bring the interaction to the most desired outcome, which is basically the maximum global welfare, or Pareto dominated equilibrium of the game. This work follows this second approach.

In this paper, we consider a network of units (PC's but possibly also smartphones or other devices possessing storage capabilities) which need to store externally a back up of their data and, at the same time, can offer space available to store data of other connected units. In this set up, we cast the peer-to-peer storage model to an allocation Potential Game and we propose an original decentralized algorithm which make units interact, cooperate, and store a complete back up of their data on their connected neighbors. 

Units are assumed to be connected through a network and, autonomously, at random time, to activate and allocate or move their data pieces among the neighboring units. Formally, each unit has a utility function which gives a value to their neighbors on the basis of their reliability, their current congestion (resources have bounded storage capabilities), and the amount of data the unit has already stored in them. Following classical evolutionary game theory \cite{SAN}, we propose an algorithm based on a noisy best response action: each time a unit activates, it decides the neighbor to use on the basis of a Gibbs probability distribution having its peak on the maxima of the utility function. 

In the remaining part of this section, we formally define the storage allocation problem and we show its equivalence with classical matching problem on a graph. This allows us to use celebrated Hall's theorem and give a necessary and sufficient condition for the allocation problem to be solvable. Section \ref{sec: algorithm} is devoted to cast the problem to a potential game theoretic framework \cite{ROS} and to propose a distributed algorithm which is an instance of a noisy best response dynamics \cite{SAN}. We claim a fundamental result, which will be proven in a forthcoming paper, which says that in the double limit when time goes to infinity and the noise parameter goes to $0$, the algorithm converges to a Nash equilibrium which is, in particular, a global maximum of the potential function. This guarantees that the solution will indeed be close to the global welfare of the community. Finally, Section \ref{sec: simulations} is devoted to the presentation of an extensive set of simulations. A conclusions section ends the paper.

\subsection{Related Work}
Though allocation games have been considered before in the literature, they all substantially differ from the model we propose and study in this paper. In \cite{ACG} the authors consider an allocation problem casted to a pure congestion game. Utility functions of units measure the congestion of a resource simply as a function of the number of units currently using it, but they do not impose any strict storage limitation. The algorithm they propose is a classical best response algorithm and is shown to achieve Nash equilibrium. Our model differs considerably as we also consider reliability of the resources and data fragmentation in the utility functions and, moreover, we impose strict storage limitations. A crucial consequence of this is that classical best response algorithms would not work in our case: Example \ref{counterexample} shows a situation where such an algorithm would halt before allocation is completed. Allocation games are also considered in \cite{CRAG} where, however, the proposed algorithm units are not interacting through a graph but rather through a device that acts like a leader selecting which resources can be used. A related context where congestion games have been used is that of networking routing \cite{HB}.

At a broader level, the noisy best response algorithm we propose in this paper fits in the so-called evolutionary game theory \cite{SAN} which has already been extensively used in studying other networking problems \cite{TEM}.

\subsection{The model}

Consider a set $\mc X$ of units which play the double role of users who have to allocate externally a back up of their data, as well resources where data from other units can be allocated. Generically, an element of $\mc X$ will be called a unit, while the terms user and resource will be used when the unit is considered in the two possible roles of, respectively, a source or a recipient of data. We assume units to be connected through a directed graph $\mc G=(\mc X,\mc E)$ where a link $(x,y)\in\mc E$ means that unit $x$ is allowed to storage data in unit $y$. We denote by $$N_x:=\{y\in\mc X\,|\, (x,y)\in\mc E\},\quad N^-_y:=\{x\in\mc X\,|\, (x,y)\in\mc E\}$$
respectively, the out- and the in-neighborhood of a node. Note the important different interpretation in our context: $N_x$ represents the set of resources available to unit $x$ while $N^-_y$ is the set of units having access to resource $y$. If $D\subseteq \mc X$, we put $N(D)=\cup_{x\in D}N_x$ and $N^-(D)=\cup_{x\in D}N_x^-$.

We imagine the data possessed by the units to be quantized atoms of the same size. Each unit $x$ is characterized by two non negative integers:
\begin{itemize}
\item $\alpha_x$ is the number of data atoms that unit $x$ needs to back up into his neighbors, 
\item $\beta_x$ is the number of data atoms that unit $x$ can accept and store from his neighbors.
\end{itemize}
The numbers $\{\alpha_x\}$ and $\{\beta_x\}$ will be assembled into two vectors denoted, respectively, $\alpha$ and $\beta$. We also define 
$$\mc A_x=\{(x,a)\,|\, a\in\{1,\dots ,\alpha_x\}\},\quad \mc A=\bigcup_{x\in\mc X}\mc A_x$$
Given the triple $(\mc G, \alpha, \beta)$, we define an allocation 
as any map $Q:\mc A\to\mc X$ satisfying the properties expressed below.
\begin{enumerate}
\item[(C1)] {\bf Graph constraint} $Q(x,a)\in N_x$ for all $x\in\mc X$ and $a\in \{1,\dots ,\alpha_x\}$;
\item[(C2)] {\bf Storage limitation} For every $y\in\mc X$, 
$$ |Q^{-1}(y)|\leq \beta_y$$
\end{enumerate}
The fact that $Q(x,a)=y$ means that user $x$ has allocated the data atom $a$ into resource $y$.

We will say that the allocation problem is solvable if an allocation $Q$ exists. We denote by $\mc Q$ the set of allocations.  We will also need to consider partial allocations, namely maps $Q:D\to\mc X$ where 
$D\subseteq \mc A$ satisfying, where defined, conditions (C1) and (C2). We denote by $\mc Q_p$ the set of partial allocations. 
%We will often use the symbol $(D,Q)$ to denote a partial $\Delta$-allocation with domain $D$, whenever we need to explicit the domain. 

In the following section we study the conditions under which the allocation problem is solvable, namely conditions under which $\mc Q$ is non empty.

\subsection{The allocation problem as a matching problem}
Define
$$\mc B_y=\{(y,b)\,|\, b\in\{1,\dots ,\beta_y\}\},\quad \mc B=\bigcup_{y\in\mc X}\mc B_y$$
Consider now the bipartite graph $\mc P=(\mc A\times \mc B, \mc E_{\mc P})$ where $((x,a),(y,b))\in\mc E_{\mc P}$ iff $(x,y)\in \mc E$. An allocation naturally induces a matching on $\mc P$ which is complete on $\mc A$. To this aim, notice that, from $Q\in\mc Q$ and using condition (C2), we can construct an injective mapping $\tilde Q:\mc A\to\mc B$ such that $\tilde Q(x,a)=(Q(x,a), b)$ for every $x\in\mc X$ and for all $a$.
%first we identify $\mc A_x$ with $\{1,\dots ,\alpha_x\}$ and we consider, for every $x\in\mc X$, mappings $\tilde Q_{x}:\mc A_x\to \mc B$ such that, $\tilde Q_{x}(a)\in \mc B_y$ iff $Q_x(a)=y$ and such that,
%for any $y\in \mc X$, as $x\in N_y^-$, the subsets $\tilde Q_{x}(\mc A_x)\cap \mc B_y$ are disjoint (this is possible because of condition (C2)). 
We then define
$$\mc M:=\bigcup\limits_{x\in\mc X}\{((x,a),(y,b))\in \mc A\times \mc B\,|\,(y,b)=\tilde Q(x,a)\}$$ 
It is clear that this procedure can be inverted and that from any matching of $\mc M$ complete on $\mc A$ we can associate an allocation for $(\mc G, \alpha, \beta)$. This equivalence allows to use classical results like the Hall's marriage theorem to characterize the existence of allocations. Precisely we have the following result
\begin{theorem}\label{theo allocation exists} Given $(\mc G, \alpha, \beta)$, there exists an allocation iff the following condition is satisfied:
\beq\label{cond allocation exists}
\sum\limits_{x\in D}\alpha_x\leq \sum\limits_{y\in N(D)}\beta_y\quad\forall D\subseteq \mc X
\eeq
\end{theorem} 
\begin{proof}
By Hall's theorem, the existence of a matching in $\mc P$ complete on $A$ is equivalent to the condition
\beq\label{Hall theorem}
|A|\leq |N^{\mc P}(A)|\quad\forall  A\subseteq \mc A
\eeq
where $N^{\mc P}(A)\subseteq B$ is the out-neighborhood of $A$ in $\mc P$.
Given $A\subseteq \mc A$ let $\bar{A}$ be the union of those $\mc A_x$'s for which $\mc A_x\cap A\neq\emptyset$. By the way the bipartite graph $\mc P$ has been defined, it follows that $N^{\mc P}(A)=N^{\mc P}({\bar A})$, so that it is sufficient to restrict condition (\ref{Hall theorem}) to subsets $A$ such that $\mc A_x\cap A\neq\emptyset$ yield $\mc A_x\subseteq A$. Given such an $A$, if we consider $D=\{x\;|\; \mc A_x\subseteq A\}$, we immediately obtain that (\ref{cond allocation exists}) coincides with (\ref{Hall theorem}).
\end{proof}

In general, it is not necessary to check the validity of (\ref{cond allocation exists}) for every subset $D$. We say that $D\subseteq \mc X$ is maximal if for any $D'\supsetneq D$, it holds $N(D')\supsetneq N(D)$. We say that $D_1, D_2\subseteq \mc X$ are independent if $N(D_1)\cap N(D_2)=\emptyset$ and $D\subseteq \mc X$ is called irreducible if it can not be decomposed into the union of two non empty independent subsets. Clearly, it is sufficient to verify (\ref{cond allocation exists}) for the subclass of maximal irreducible subsets.

\begin{example} If $\mc G$ is complete, we have that $N(\{x\})=\mc X\setminus\{x\}$ while $N(D)=\mc X$ for all $D$ such that $|D|\geq 2$. Hence, the only maximal irreducible subsets are the singletons $\{x\}$ and the set $\mc X$. Condition (\ref{cond allocation exists}) in this case reduces to
\beq\label{cond allocation exists complete}
\alpha_x\leq \sum\limits_{y\neq x}\beta_y,\; \forall x\in\mc X\qquad 
\sum\limits_{x\in\mc X}\alpha_x\leq \sum\limits_{y\in \mc X}\beta_y
\eeq
\end{example}
In general, the class of maximal irreducible subsets can be large and grow more than linearly in the size of $\mc X$, as the following example shows.
\begin{example} If $\mc G=(\mc X, \mc E)$ is a line graph ($\mc X=\{1,2,\dots , n\}$ and $\mc E=\{(i, i+1), \,i=1, \dots , n-1\}$) it can be checked that the maximal irreducible subsets are those of the form $\{i, i+2, \dots , i+2s\}$. 
\end{example}

In case when $\alpha_x=a$ and $\beta_y=b$ are both constant, something more can be said.
\begin{proposition}\label{prop allocation}  Given $(\mc G, \alpha, \beta)$, where $\mc G$ is regular, $\alpha_x=a$ and $\beta_y=b$ for all $x, y\in\mc X$, there exists an allocation iff $a\leq b$.
\end{proposition}
\begin{proof} By Theorem \ref{theo allocation exists}, we simply have to show that $|D|\leq |N(D)|$ for every subset $D\subseteq \mc X$. Let $\mc E_{D}$ be the set of edges having one of the nodes in $D$. If $d$ is the degree of the nodes in $\mc G$, we have that
$$d|D|=|\mc E_{D}|\geq d|N(D)|$$
\end{proof}

From the practical point of view, the equivalence of our problem with a classical matching problem, is, however, of little utility, as the number of nodes of $\mc P$ is of the size $\sum\alpha_x+\sum\beta_y$ which will in general be very large. 

Moreover, in case allocations exist, we want to be able to construct one in a distributed way without the need of any supervision. The algorithm must be iterative in order to cope with possible time modifications of the units, of their interconnection and of their data storage needs and limitations. Also the possibility that units leave and enter the community must be considered. 

Also we want to have the possibility to find solution possessing certain extra features:
\begin{itemize}
\item {\bf (Reliability)} resources will often have different reliability properties and we want to give preference, in the allocation, to more reliable resources;
\item {\bf (Congestion)} equally reliable resources should be equally used, avoiding congestion phenomena;
\item {\bf (Aggregation)} users prefer to use as few resources as possible to allocate their back up data.
\end{itemize}
The reason for this last feature comes from the fact that an exceeding fragmentation of the back up data will cause a blow up in the number of communications among the units both in the storage and recover phases. This feature should be considered against another feature which in this paper is not going to be addressed, which is that of diversification of back ups: in real applications units will need to back up multiple copies of their data in order to cope with security and possible failure phenomena. In that case, these multiple copies will need to be stored in different units. This issue will be analyzed in a subsequent paper.

The above desired features may be contradictory in general and we want to have tunable parameters to make the algorithm converge towards a desired compromised solution. 

The proposed algorithm will be fully distributed: units will iteratively allocate and move their data among the neighbors on the basis of the space available and trying to maximize a utility function. There will be an underlying game theoretic structure inspired by the desired features described above. Our algorithm will be analyzed with the techniques of evolutionary game theory and it will be shown to yield a reversible Markov process converging to a Nash equilibrium of the game.

\section{The game theoretic set-up and the algorithm}\label{sec: algorithm}
Given a partial allocation $Q\in\mc Q_p$, consider the matrix $W(Q)\in\N^{\mc X\times \mc X}$ where $W(Q)_{xy}$ is the number of atomic data that $x$ has copied inside $y$ under the allocation $Q$, namely,
\begin{equation}\label{W}W(Q)_{xy}:=|Q^{-1}(y)\cap\mc A_x|\end{equation}
Clearly $W=W(Q)$ satisfies the following conditions
\begin{enumerate}
\item[(P1)]  $W_{xy}\geq 0$ for all $x,y$ and $W_{xy}=0$ whenever $(x,y)\not\in \mc E$.
\item[(P2)]  $W^x:=\sum\limits_{y\in\mc X}W_{xy}\leq\alpha_x$ for all $x\in\mc X$.
\item[(P3)] $W_y:=\sum\limits_{x\in\mc X}W_{xy}\leq \beta_y$ for all $y\in\mc X$.
\end{enumerate}
It is immediate to see that, conversely, if there exists $W$ satisfying these properties (such a $W$ is called a partial allocation state), then, from it, we can construct a partial allocation $Q$ such that $W=W(Q)$. Clearly, under this correspondence, we have that $Q\in \mc Q$ iff $W$ satisfies (P2) with equality for all $x\in\mc X$. In this case $W$ is called an allocation state.
The set of partial allocation states and the set of allocation states are denoted, respectively, with the symbols $\mc W_p$ and $\mc W$.

It is clear that two partial allocations $Q^1$ and $Q^2$ such that $W(Q^1)=W(Q^2)$, only differ for a permutation of the data atoms of the various units and for many purpouses can be considered  as equivalent. All the quantity of interest for the game theoretic setting will be defined at the level of $W$.

We are now ready to define the game theoretic model. We first define utilities: under a (possibly partial) allocation state $W$, the utility of a unit $x$ in using resource $y$ is given by
\beq\label{payoff units}
f_{xy}(W):=\lambda_y- k_cW_y/\beta_y+k_aW_{xy} 
\eeq
The first term $\lambda_y$ encodes possible reliability differences among resources, the second term is instead a congestion term which takes into consideration the level of use of the resource, and, finally, the third term depends on both $x$ and $y$ and pushes a unit to allocate in those resources where it has already allocated. $k_c, k_a$ are two non-negative parameters to tune the effect of the congestion and of the aggregation terms, respectively.

The choice of this particular utility function has been made on the basis of simplicity considerations (notice that the state $W$ enters linearly in it) and on the fact that, as exploited below, this leads to a potential game. In principle, different terms in the utility function can be introduced in order to make units to take into considerations other desired features (e.g multiple back up).

Define
\beq\label{potential discrete W} \Psi(W):=\sum\limits_{y\in\mc X}\sum\limits_{s=0}^{W_y}[\lambda_y- k_cs/\beta_y]+k_a\sum\limits_{x,y\in\mc X}\sum\limits_{s=0}^{W_{xy}}s\eeq
and notice that if $W, W'\in\mc W$ are such that $W'=W-e_{x\bar y}+e_{x\bar y'}$, then,
\beq\label{potential relation}\begin{array}{l}\Psi(W')-\Psi(W)\\[8pt]
%&=&[\lambda_{\bar y}- K_cW(Q)_{\bar y}/\beta_{\bar y}]-[\lambda_{\bar y'}- K_cW(Q)_{\bar y'}/\beta_{\bar y'}]-\\[8pt]
%&+&\frac{K_a}{2} W(Q)_{\bar x\bar y}(W(Q)_{\bar x \bar y}+1)-\frac{K_a}{2} (W(Q)_{\bar x\bar y}-1)W(Q)_{\bar x \bar y}\\[8pt]

%&+&\frac{K_a}{2} (W(Q')_{\bar x\bar y'}-1)W(Q')_{\bar x \bar y'}-\frac{K_a}{2} W(Q')_{\bar x\bar y'}(W(Q')_{\bar x \bar y}+1)\\[8pt]

=(\lambda_{\bar y'}- k_cW'_{\bar y'}/\beta_{\bar y'}+k_aW'_{\bar x\bar y'})-(\lambda_{\bar y}- k_cW_{\bar y}/\beta_{\bar y}+k_aW_{\bar x\bar y})\\[8pt]
=f_{x\bar y'}(W')-f_{x\bar y}(W)\end{array}\eeq
In other terms, under the state allocation $W$, when user $x$ moves a data atom from $\bar y$ to $\bar y'$, it experiences a variation in utility given by $\Psi(W')-\Psi(W)$. $\Psi$ is called a potential of the game.
Given $W\in\mc W_p$ and $x\in\mc X$, put 
$$\mc X^x(W):=\{y\in N_x\;|\; W_y<\beta_y\}$$
the set of resources still available for $x$ under the current allocation state $W$. 

An allocation state $W\in\mc W$ (and also any $Q\in\mc Q$ such that $W(Q)=W$) is called a Nash equilibrium if, for every $x\in\mc X$, for every $y\in N_x$ such that $W_{xy}>0$, for every $y'\in \mc X^x(W)$, it holds
$$f_{xy}(W)\geq f_{xy'}(W')$$
Maxima of $\Psi$ are clearly Nash equilibria while, in general, the converse is not true. Considering that $\Psi$ is defined on a finite set, a maximum, and thus a Nash equilibrium, always exists.
Under a Nash equilibrium, a unit whose goal is to maximize its utility, has no advantage in moving their allocated data, under the standing assumption that only one data atom at a time can be moved. Notice that data atoms are to be interpreted as aggregations of data and the decision of their size is part of the design of the algorithm. Clearly different levels of granularity will give rise to different game models including different Nash equilibria. 

For any $W\in\mc W_p$, we define the Gibbs probability distribution over $\mc X^x(W)$as
$$p_y(W,x)=\frac{e^{\gamma f_{xy}(W+e_{xy})}}
{Z_{\gamma}}$$
where $\gamma>0$ and where 
$$Z_{\gamma}=\sum_{y\in\mc X^x(W)}e^{\gamma f_{xy}(W+e_{xy})}$$ is a normalizing factor.

\subsection{The algorithm}

The algorithm we are proposing below is a distributed and asynchronous algorithm where units activate at random independent times and either allocate or move their atoms, undertaking a relaxed stochastic version of the utility maximization. 

The algorithm is mathematically described as a continuous time Markov process $Q_t:D_t\to \mc X$ on the set of partial allocations $\mc Q_p$.
%$W(t)$ over the set $\mc W_p$ of partial allocation states. 
Precisely, units are assumed to be equipped with independent internal Poisson clocks with possibly different clicking rates. We denote by $\nu_x$ the clicking rate of unit $x$. When a unit activates it can either allocate a further data atom (if allocation is not completed yet) or move a data atom from one resource to another. The choice of the resource where either allocate or move the data atom is done according to the Gibbs probability: this is a classical choice in evolutionary game theory and will be amenable to a fairly complete theoretical analysis.
The details of the algorithm are described below. We put $W(t)=W(Q_t)$.

\begin{enumerate}
%\item Assume that $W(t)$ is the (partial) allocation state at time $t$
\item Assume $\bar x$ activates at time $t$. With probabilities 
$$P_{\rm all}(W(t), \bar x),\quad  P_{\rm dis}(W(t),\bar x)$$ the resource $\bar x$ will make, respectively, an allocation or a distribution move as explained below. Of course we are assuming that 
$$\begin{array}{l}P_{\rm all}(W(t),\bar x)+P_{\rm dis}(W(t),\bar x)=1,\\[5pt]
P_{\rm all}(W(t),\bar x)=0 \; {\rm if}\; W(t)^{\bar x}=\alpha_{\bar x}\\[5pt]
P_{\rm dis}(W(t),\bar x)=0 \; {\rm if}\; W(t)^{\bar x}=0\end{array}$$

\item (ALLOCATION MOVE)
\begin{itemize}
\item Choose $(\bar x,\bar a)$ uniformly at random in $(\mc A\setminus D_t)\cap\mc A_{\bar x}$
\item  Choose $y^*$ according to the Gibbs probability $p_{y^*}(W(t), \bar x)$
\item Put $D_{t+}=D_t\cup\{(\bar x,\bar a)\},$ and $Q_{t^+}:D_{t+}\to\mc X$ by
$$\quad Q_{t^+}(x,a)=\left\{\begin{array}{ll} y^*\;\; &{\rm if}\,(x,a)=(\bar x,\bar a)\\
Q_t(x,a)\;\; &{\rm if}\, (x,a)\in D_t\end{array}\right.$$
%$W(t+)=W(t)+e_{xy^*}$
\end{itemize}
\item (DISTRIBUTION MOVE)
\begin{itemize}
\item Choose $\bar y$ according to the probability 
$q_{\bar y}=W(t)_{\bar x \bar y}/W(t)^{\bar x}$.
\item Choose $(\bar x,\bar a)$ uniformly at random in $D_t\cap\mc A_{\bar x}$
\item Choose $y^*$ according to the Gibbs probability $p_{y^*}(W(t)-e_{\bar x\bar y}, \bar x)$
\item Put $D_{t+}=D_t,$
$$\quad  Q_{t^+}(x,a)=\left\{\begin{array}{ll} y^*\;\; &{\rm if}\,(x,a)=(\bar x,\bar a)\\
Q_t(x,a)\;\; &{\rm if}\, (x,a)\in D_t\setminus \{(\bar x,\bar a)\}\end{array}\right.$$
%$W(t+)=W(t)-e_{xy}+e_{xy^*}$
\end{itemize}
\end{enumerate}

The fact of using a noisy algorithm is crucial in our setting. In the following example we show a situation where a classical best response algorithm would remain stacked without completing the allocation.
\begin{example}\label{counterexample}
We are considering a line graph of four users as depicted below. 
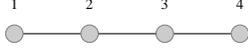
\begin{center}
\begin{tikzpicture}
  [scale=1,auto=left,every node/.style={circle,draw=black!50,fill=black!20,scale=.7}]
  
  \node (n1) at (0,0) {};
  \node (n2) at (1,0)  {};
  \node (n3) at (2,0)  {};
  \node (n4) at (3,0)  {};

 \draw node[draw=white,fill=white,scale=.8] at (0,0.4) {$1$};
 \draw node[draw=white,fill=white,scale=.8] at (1,0.4) {$2$};
 \draw node[draw=white,fill=white,scale=.8] at (2,0.4) {$3$};
 \draw node[draw=white,fill=white,scale=.8] at (3,0.4) {$4$};

 \foreach \from/\to in
  {n1/n2,n2/n3, n3/n4}
    \draw (\from) -- (\to);

\end{tikzpicture}   
\end{center} 
\vspace{.2cm}
Each user $x$ has $\alpha_x=\beta_x=1$. Reliabilities are instead $\lambda_2=3$ while $\lambda_x=1$ for $x=1,3,4$. Assume we are in the partial allocation state $W\in\mc W_p$ given by
$$W:=\left(\begin{matrix}0& 0 & 0 &0\\ 1 &0 &0 &0\\ 0 &1 &0 &0\\ 0 &0 &1 &0\end{matrix}\right)$$ 
Clearly, this partial allocation state could be reached by the group zero allocation state with positive probability (the key point is that unit $3$ activates before unit $1$ and chooses the most reliable resource). It is also clear that under a best response strategy ($\gamma=+\infty$) this state allocation is an equilibrium: unit $1$ can not allocate, unless unit $3$ moves its data to $4$ but this will never happen because 
$f_{34}(W-e^{32})<f_{32}(W-e^{32})$.

%where each user has $\alpha_x=1$ and $\beta=1$ but $\lambda_2=3$ while $\lambda_i=1$ for $i=1,3,4$. When user $3$ plays, he choose resource $2$ which is the most reliable but now user $1$ does not have the possibility to allocate because his only available resource is full. With he classical best response algorithm the game end in this situation and never exit because user $3$ is doing his best strategy; adding the noise, in the long run there is a possibility for $3$ to change the resource he is using allowing $1$ to allocate his atom and reaching the Nash equilibrium.

\end{example}

\subsection{Theoretical results}
In this section we analyze the behavior of the algorithm introduced above. We will essentially show two results. First, we prove that if the set of allocation $\mc Q$ is not empty (i.e. condition (\ref{cond allocation exists}) is satisfied), the algorithm above will find one in bounded time with probability $1$. Second, we will show that, under a slightly stronger assumption than (\ref{cond allocation exists}), in the double limit $t\to +\infty$ and then $\gamma\to +\infty$, the Markov process induced by the algorithm will always converge, in law, to a Nash equilibrium which is a global maximum of the potential function $\Psi$. 

In order to prove such results, it will be necessary to go through a number of intermediate technical steps. First of all, it will be convenient to work directly with the process $W(t)=W(Q(t))$ which is also Markovian because of the way the transitions have been defined, considering that the results we are claiming, can all be expressed and established at this simpler level.

Given $W,W'\in\mc W_p$, we denote by $P_{W,W'}$ the transition probability from $W$ to $W'$ of the Markov chain underlying the process $W(t)$. $\mc L_p$ denotes the graph on $\mc W_p$ where an edge $(W,W')$ is present if and only if $P_{W,W'}>0$.

Our strategy will be to show that from any element $W\in\mc W_p$ there is a path in $\mc L_p$ to some element $W'\in \mc W$. This, by standard Markov chain arguments, leads to the result that allocation will be achieved in bounded time with probability $1$. After, we will show that $\mc L_p$ restricted to the set of allocations $\mc W$ is irreducible and this will then yield the asymptotic result.

First we consider the problem of finding a path to an allocation state. Given $W\in \mc W_p$ we define the following subsets of units 
$$\mc X^f(W):=\{x\in\mc X\;|\; W^x=\alpha_x\}\,$$ 
$$\mc X^{sat}(W):=\{x\in\mc X\setminus\mc X^f(W)\;|\; \not\exists y\in N_x\;\hbox{\rm s.t}\; W(Q)_y< \beta_y\}$$
%Then, we define
%$$\mc X^{sat}(W):=\{x\in\mc X\setminus\mc X^f(W)\;|\; \not\exists y\in N_x\;\hbox{\rm s.t}\; W(Q)_y< \beta_y\}$$
Units in $\mc X^f(W)$ are called fully allocated: these units have completed the allocation of their data under the state $W$. Units in $\mc X^{sat}(W)$ are called saturated: they have not yet completed their allocation, however, under the current state $W$, they can not make any action, neither allocate, nor distribute. 
%We call the units in $\mc X^{sat}(D,Q)$ saturated: are those units who have not yet allocated all their data, but no further allocation or redistribution can be done under the current $\Delta$-allocation $Q$.
%Notice indeed that if $x\in \mc X^{sat}(D,Q)$, then
%$$\mc X^{(x,s,j)}(Q)=\left\{\begin{array}{ll}\emptyset \;\;&{\rm if}\; (x,s,j)\not\in D\\
%\{Q_{x,s,j}\} \;\;&{\rm if}\; (x,s,j)\in D\end{array}\right.$$
%so that when the atom $(x,s,j)$ activates, the Markov process will remain in $(D,Q)$.
Finally, define
$$\mc W_p^{sat}:=\{W\in\mc W_p\setminus\mc W\;|\; \mc X=\mc X^f(W)\cup \mc X^{sat}(W)\}$$
It is clear that from any $W\in\mc W_p\setminus \mc W_p^{sat}$, some allocation move can be performed. Instead, if we are in a state $W\in \mc W_p^{sat}$, only possibly fully allocated units can make a distribution move. Notice that, because of condition (\ref{cond allocation exists}), for sure there exist resources $y$ such that $W_y<\beta_y$ and these resources are indeed exclusively connected to fully allocated units. The key point is to show that in a finite number of distribution moves it is always possible to move some data atoms from resources connected to saturated units to resources with available space: this will then make possible a further allocation move.

For any fixed $W\in W_p$, we can consider the following graph structure on $\mc X$ thought as set of resources: $\mc H_W=(\mc X,\mc E_W)$. Given $y_1,y_2\in\mc X$, there is an edge from $y_1$ to $y_2$  if and only if there exists $x\in\mc X$ for which
$$W_{xy_1}>0,\quad (x,y_2)\in\mc E$$
The edge from $y_1$ to $y_2$ will be indicated with the symbol $y_1\to_x y_2$ (to also recall the unit $x$ involved). The presence of the edge means that the two resources $y_1$ and $y_2$ are in the neighborhood of a common unit $x$ which is using $y_1$ under $W$. This indicates that $x$ can in principle move some of its data currently stored in $y_1$ into resource $y_2$ if this last one is available. We have the following technical result

\begin{lemma}\label{lemma equilibrium 1} Suppose $(\mc G ,\alpha, \beta)$  satisfies (\ref{cond allocation exists}). Fix $W\in W_p$ and let $\bar y\in\mc X$  be such that there exists $\bar x\in N_{\bar y}$ with $W^{\bar x}<\alpha_{\bar x}$.
Then, there exists a sequence
\beq\label{sequence} \bar y=y_0,\,x_0,\,y_1,\dots ,y_{t-1},\,x_{t-1},\,y_t\eeq
satisfying the following conditions
\begin{enumerate}
\item[(Sa)] Both families of the $y_k$'s  and of the $x_k$'s are each made of distinct elements;
\item[(Sb)] $y_k\to_{x_k} y_{k+1}$ for every $k=0,\dots ,t-1$;
%(i.e. $W_{y_kx_k}>0$, $(x_k,y_{k+1})\in\mc E$, ${W(Q)}_{x_ky_{k+1}}<\alpha_{x_k}$);
\item[(Sc)] ${W}_{y_k}=\beta_{y_k}$ for every $k=0,\dots ,t-1$, and ${W(Q)}_{y_t}<\beta_{y_t}$.
\end{enumerate}
\end{lemma}
\begin{proof}
Let $\mc Y\subseteq \mc X$ be the subset of nodes which can be reached from $\bar y$ in $\mc H_W$.
Preliminarily, we prove that there exists $y'\in\mc Y$ such that $W_{y'}<\beta_{y'}$.
Let
$$\mc Z:=\{x\in\mc X\;|\; \exists y\in\mc Y,\, W_{xy}>0\}$$
and notice that, by the way $\mc Y$ and $\mc Z$ have been defined,
\beq\label{condZ} x\in\mc Z,\; (x,y)\in \mc E\;\Rightarrow\; y\in\mc Y\eeq
Suppose now that, contrarily to the thesis, $W_y\geq \beta_y$ for all $y\in\mc Y$. Then, 
\beq\label{estim}\begin{array}{rcl}\sum\limits_{x\in \mc Z}\alpha_x &\leq& \sum\limits_{y\in \mc Y}\beta_y\\[16pt]
&=&\sum\limits_{y\in \mc Y}W_y\\[16pt]
&=&\sum\limits_{y\in \mc Y}\sum\limits_{x\in\mc Z}W_{xy}\\[12pt]
&=&\sum\limits_{x\in\mc Z}W^x\\[12pt]
&<&\sum\limits_{x\in \mc Z}\alpha_x 
\end{array}
\eeq
where the first inequality follows from (\ref{condZ}) and (\ref{cond allocation exists}), the first equality from the contradiction hypothesis, the second equality from the definition of $\mc Z$, the third equality again from (\ref{condZ}) and, finally, last inequality from the existence of $\bar x$. This is clearly absurd and thus proves our claim.

Consider now a path of minimal length from $\bar y$ to $\mc Y$ in $\mc H_W$:
$$\bar y=y_0\to_{x_0}y_1\to\dots \to y_{t-1}\to_{x_{t-1}}y_t$$
and notice that the sequence $\bar y=y_0,\,x_0,\,y_1,\dots ,y_{t-1},\,x_{t-1},\,y_t$ will automatically satisfy properties (Sa) to (Sc).
\end{proof}

We are now ready to prove the first main result.

\begin{theorem} \label{theo main 1} Assume that 
\begin{enumerate}
\item $\nu_x>0$ for every $x\in\mc X$ such that $\alpha_x>0$, 
\item $P_{\rm all}(W, \bar x)>0$ if $W^{\bar x}<\alpha_{\bar x}$, 
\item $(\mc G ,\alpha, \beta)$  satisfies (\ref{cond allocation exists}). 
\end{enumerate}
Then, with probability $1$, the Markov process $W(t)$ will be, after a finite number of jumps, in the set of allocations $\mc W$.
\end{theorem}
\begin{proof}
In order to prove the claim, it will be sufficient to show that from any $W\in\mc W_p$ there is a path in $\mc L_p$ (the graph underlying the possible  transitions of the process $W(t)$)  to some element $W'\in \mc W$. We will prove it by a double induction process. To this aim we consider two indices associated to any $W\in\mc W_p\setminus \mc W$. The first one is defined by
$$m_W=\sum\limits_{x\in\mc X}(\alpha_x-W^x)\geq 1$$
To define the second, consider any  $\bar x\in \mc X\setminus \mc X^f(W)$. We can apply Lemma \ref{lemma equilibrium 1} to $W$ and any $\bar y\in N_{\bar x}$ and obtain that we can find a sequence of agents 
$\bar y=y_0,\,x_0,\,y_1,\dots ,\,y_{t-1},\,x_{t-1},\,y_t$
satisfying the properties (Sa), (Sb), and (Sc) above. Among all the possible choices of $\bar x\in\mc X$, $\bar y\in\ N_{\bar x}$ and of the corresponding sequence, assume we have chosen the one minimizing $t$ and denote such minimal $t$ by $t_W$. The induction process will be performed with respect to the lexicographic order induced by the pair $(m_W, t_W)$. 

In the case when $t_W=0$, it means we can find $\bar x\in\mc X$, $\bar y\in\ N_{\bar x}$ such that $W_{\bar y}<\beta_{\bar y}$. Therefore, under the allocation state $W$, the unit $\bar x$ can allocate a further data atom to $\bar y$. Considering that the activation of the unit $\bar x$ has positive probability because of assumptions 1. and 2., this shows that $W$ is connected to a $W'$ such that $m_{W'}<m_W$. In case $m_W=1$, this means that $W'\in\mc W$. 

Consider now any $W\in\mc W_p\setminus \mc W$ such that $t=t_W>1$. Let $\bar x\in\mc X$, $\bar y\in\ N_{\bar x}$  and the sequence $\bar y=y_0,\,x_0,\,y_1,\dots ,\,y_{t-1},\,x_{t-1},\,y_t$
satisfying the properties (Sa), (Sb), and (Sc) above. In the allocation state $W$, the unit $x_{t-1}$, if activated (and this again has positive probability to happen because of assumptions 1. and 2.), can thus move an atomic piece of data from $y_{t-1}$ to $y_t$. The new allocation state is $W'=W-e_{x_{t-1}y_{t-1}}+e_{x_{t-1}y_t}$. Since $W'_{y_{t-1}}<\beta_{y_{t-1}}$, for sure $t_{W'}<t_W$. The induction argument is thus complete. 

\end{proof}

We are now left with studying the Markov process $W(t)$ on $\mc W$.  We start with the following
\begin{proposition}\label{prop time rev} $W(t)$, restricted to $\mc W$, is a time reversible Markov process. More precisely, it holds
\beq\label{time reversible}{\alpha\choose W}e^{\gamma \Psi(W)}P_{W,W'}= {\alpha\choose W'} e^{\gamma \Psi(W')}P_{W',W}\quad\forall W,W'\in\mc W
\eeq
where $${\alpha\choose W}:=\displaystyle\frac{\prod\limits_x\alpha_x!}{\prod\limits_{x,y}W_{xy}!}$$
\end{proposition}
\begin{proof}
Notice that the only cases when $P_{W,W'}$ and $P_{W',W}$ are not both equal to $0$ is for those pairs $W, W'\in\mc W$ such that $W'=W-e_{x\bar y}+e_{x\bar y'}$ for some $x, \bar y, \bar y'\in\mc X$. Assume this to be the case. Then, it follows from the way the distribution moves of the algorithm have been defined that
$$P_{W,W'}=\nu_x\frac{W_{x\bar y}}{\alpha_x}e^{\gamma f_{x\bar y'}(W')},\; P_{W',W}=\nu_x\frac{W'_{x\bar y'}}{\alpha_x}e^{\gamma f_{x\bar y}(W)}$$
Substituting in (\ref{time reversible}) and using relation (\ref{potential relation}), it is immediate to check that equality holds.

\end{proof}

We now show that under a slight stronger assumption than (\ref{cond allocation exists}), namely,
\beq\label{cond W ergodic}
\sum\limits_{x\in A}\alpha_x< \sum\limits_{y\in N(A)}\beta_y\quad\forall A\subseteq \mc X\,,
\eeq
the process $W(t)$ restricted to $\mc W$ is ergodic. Denote by $\mc L$ the graph $\mc L_p$ restricted to the set $\mc W$. Notice that, as a consequence of time-reversibility, $\mc L$ is an undirected graph. Ergodicity is equivalent to proving that $\mc L$ is connected. We start with a lemma analogous to previous Lemma \ref{lemma equilibrium 1}.

\begin{lemma}\label{lemma equilibrium 2} Suppose $(\mc G ,\alpha, \beta)$  satisfies (\ref{cond W ergodic}) and let $W\in \mc W$. Then, for every $\bar y\in\mc X$, there exists a sequence (\ref{sequence})
satisfying the conditions (Sa), (Sb), and (Sc) as in Lemma \ref{lemma equilibrium 1}. 
\end{lemma}
\begin{proof} It is sufficient to follow the steps of to the proof of Lemma \ref{lemma equilibrium 1} noticing that in (\ref{estim}) the first equality is now a strict inequality, while the last strict inequality becomes an equality. 
\end{proof}

If $W,W'\in\mc W$ are connected through a path in $\mc L$, we write that $W\sim W'$. Introduce the following distance on $\mc W$: if $W^1, W^2\in\mc W$
$$\delta(W^1,W^2)=\sum\limits_{x,y}|W^1_{xy}-W^2_{xy}|$$
A pair $(W^1, W^2)\in\mc W$ is said to be minimal if 
$$\delta (W^1,W^2)\leq \delta (W^{1'}, W^{2'})\;\;\forall W^{1'}\sim W^1,\; \forall W^{2'}\sim W^2$$
Notice that $\mc L$ is connected if and only if for any minimal pair $(W^{1}, W^{2})$, it holds $W^{1}=W^{2}$. 

\begin{lemma}\label{lemma minimal1} Let $(W^1, W^2)$ be a minimal pair.
Suppose $y\in\mc X$ is such that $W^1_y<\beta_y$. Then, $W^1_{xy}=W^2_{xy}$ for all $x\in\mc X$.
\end{lemma}
\begin{proof}
Suppose by contradiction that $W^1_{xy}<W^2_{xy}$ for some $x\in\mc X$. Then, necessarily, there exists $y'\neq y$ such that $W^1_{xy'}>W^2_{xy'}$. Consider then $W^{1'}=W^1-e^{xy'}+e^{xy}$. Clearly, $\delta (W^{1'},W^2)<\delta(W^1,W^2)$ and this contradicts the minimality assumption. Thus $W^1_{xy}\geq W^2_{xy}$ for all $x\in\mc X$. This yields $W^2_y<\beta_y$. Exchanging the role of $W^1$ and $W^2$ we obtain the thesis.
\end{proof}

\begin{proposition}\label{prop connected} If condition (\ref{cond W ergodic}) holds true, the graph $\mc L$ is connected.
\end{proposition}
\begin{proof}
Let $(W^1, W^2)$ be any minimal pair. We will prove that $W^1$ and $W^2$ are necessarily identical. Consider any resource $y$. It follows from Lemma \ref{lemma equilibrium 2} that we can find a sequence $y=y_0,\,x_0,\,y_1\cdots ,\,y_{t-1},\,x_{t-1},\,y_t$ satisfying the same (Sa), (Sb), and (Sc) with respect to the state allocation $W^1$. Among all the possible sequences, choose one with $t$ minimal for given $y$. We will prove by induction on $t$ that $W^1_{xy}=W^2_{xy}$ for all $x\in\mc X$.

If $t=0$, it means that $W^1_{y}<\beta_{y}$. It then follows from Lemma \ref{lemma minimal1} that $W^1_{xy}=W^2_{xy}$ for all $x\in\mc X$. Suppose now that the claim has been proven for all minimal pairs $(W^1,W^2)$ and any $y\in\mc X$ for which $t< \bar t$ (w.r. to $W^1$) and assume that $y=y_0,\,x_0,\,y_1\cdots ,\,y_{\bar t-1},\,x_{\bar t-1},\,y_{\bar t}$ satisfyies the properties (Sa), (Sb), and (Sc) with respect to $W^1$. Notice that the unit $x_{\bar t-1}$ can move a data atom from resource $y_{\bar t-1}$ into resource $y_{\bar t}$ under the state allocation $W^1$ and obtain $W^{1'}=W^1-e^{x_{\bar t-1}y_{\bar t-1}}+e^{x_{\bar t-1}y_{\bar t}}$. Consider now $W^2$ and notice that Lemma \ref{lemma minimal1} yields $W^2_{x_{\bar t-1}y_{\bar t}}=W^1_{x_{\bar t-1}y_{\bar t}}<\beta_{y_{\bar t}}$. Define
$$W^{2'}=\left\{\begin{array}{ll}W^2\;&{\rm if}\, W^2_{x_{\bar t-1}y_{\bar t -1}}=0\\
W^2-e^{x_{\bar t-1}y_{\bar t-1}}+e^{x_{\bar t-1}y_{\bar t}}\;&{\rm if}\, W^2_{x_{\bar t-1}y_{\bar t -1}}>0\end{array}\right.$$
Clearly, $\delta(W^{1'},W^{2'})\leq \delta(W^{1},W^{2})$ and this implies that also $W^{1'},W^{2'}$ is a minimal pair. Notice that $y=y_0,\,x_0,\,y_1\cdots ,\,y_{\bar t-1}$ satisfies (Sa), (Sb), and (Sc) with respect to $W^{1'}$. Therefore, by the induction hypotheses, it follows that $W^{1'}_{xy}=W^{2'}_{xy}$ for all $x\in\mc X$. Since $W^1_{xy}=W^{1'}_{xy}$ and $W^2_{xy}=W^{2'}_{xy}$, result follows immediately.

\end{proof}

Propositions \ref{prop time rev} and \ref{prop connected} immediately yield the following final result.
\begin{corollary} \label{cor main 2} Assume that (\ref{cond W ergodic}) holds true. Then $W(t)$, restricted to $\mc W$, is an ergodic time reversible Markov process whose unique invariant probability measure is given by
$$\mu_{\gamma}(W)=\frac{{\alpha\choose W}e^{\gamma \Psi(W)}}{\sum\limits_{W'\in\mc W}{\alpha\choose W'}e^{\gamma \Psi(W')}}$$
\end{corollary}

\medskip
{\bf Remark:} Notice that when $\gamma\to +\infty$, the invariant probability $\mu_{\gamma}$ converges to the probability $\mu_{\infty}$ concentrated on the set $\argmax_{W\in\mc W} \Psi(W)$ of state allocations maximizing the potential and given by, for $W\in\argmax_{W\in\mc W} \Psi(W)$, 
$$\mu_{\infty}(W)=\frac{{\alpha\choose W}}{\sum\limits_{W'\in\mc W}{\alpha\choose W'}}$$
Thus, if $\gamma$ is small, the distribution of the process $W(t)$ for $t$ sufficiently large will be close to Nash equilibria.

%\beq\label{invariant W} \mu(W):=h(W)e^{\gamma\Psi(W)}\eeq 
%where
%\beq\label{potential discrete W} \Psi(W):=\sum\limits_{y\in\mc X}\sum\limits_{s=0}^{W_y}[\lambda_y- k_cs/\beta_y]+k_a\sum\limits_{x,y\in\mc X}\sum\limits_{s=0}^{W_{xy}}s\eeq
%and where
%$$h(W):=\prod\limits_{x\in\mc X}\frac{\alpha_x!}{\prod_yW_{xy}!}$$

%{\bf Remark:}
In this paper we will assume that $\nu_x=\nu \alpha_x$ for some $\nu>0$, namely that units activation rates is proportional to the amount of data they need to back up.
We will consider two possibilities for the allocation and distribution probabilities $P_{\rm all}(W(t),x)$ and $P_{\rm dis}(W(t),x)$:
$$\begin{array}{lcl}
P_{\rm all}(W(t),x)&=&\frac{\alpha_x-W(t)^x}{\alpha_x},\\[10pt]
P_{\rm all}(W(t),x)&=&\left\{\begin{array}{ll} 1\;&{\rm if}\, W(t)^x<\alpha_x\\ 0\; &{\rm otherwise}\end{array}\right. \\[10pt]
P_{\rm dis}(W(t),x)&=&\frac{W(t)^x}{\alpha_x},\\[10pt]
P_{\rm dis}(W(t),x)&=&\left\{\begin{array}{ll} 0\;&{\rm if}\, W(t)^x<\alpha_x\\ 1\; &{\rm otherwise}\end{array}\right.
\end{array}$$

%$$\begin{array}{rclrcl}P_{\rm all}(W(t),x)&=&\frac{\alpha_x-W(t)^x}{\alpha_x}, &P_{\rm dis}(W(t),x)&=&\frac{W(t)^x}{\alpha_x}\\[10pt]
%P_{\rm all}(W(t),x)&=&\left\{\begin{array}{ll} 1\;&{\rm if}\, W(t)^x<\alpha_x\\ 0\; &{\rm otherwise}\end{array}\right.
%, &P_{\rm dis}(W(t),x)&=&\left\{\begin{array}{ll} 0\;&{\rm if}\, W(t)^x<\alpha_x\\ 1\; &{\rm otherwise}\end{array}\right.\end{array}
%$$
In the first case, the probability of an allocation move is proportional to the amount yet to be allocated while in the second case, no distribution move takes place before full allocation is reached.

\section{Simulation}\label{sec: simulations}
In this section we present a number of numerical simulations that validate the theoretical results and show the performance of the algorithm in terms of various parameters describing the speed of convergence, resources congestions, global utility, and complexity of the interconnections.

For the sake of readability we gather below the standing assumptions and parameters we have used.
\begin{itemize}
\item 
The number of units is denoted by $n$ and assumed to be even. Most of our simulations are for $n=50$ but we have also studied scalability issues by considering $n=100$ and $n=1000$.
\item We have considered two possible interconnection topologies: the complete graph and a random regular graph of degree $10$. 
\item Time has been assumed to be discrete, assuming that at every time instant a unit $x$ is chosen with a probability $\nu_x$ proportional to $\alpha_x$. Moreover, allocation and distribution moves are chosen acceding to the probabilities:
$$P_{\rm all}(W(t),x)=\left\{\begin{array}{ll} 1\;&{\rm if}\hspace{.2cm} W(t)^x<\alpha_x\\ 0\; &{\rm otherwise}\end{array}\right.
$$
\item Time horizon has been fixed to be $T=2\sum\alpha_x$ so that we have allowed up to two moves per data atom (one allocation and one possible distribution). It turns out that in all experiments carried on such time horizon has been sufficient for completing the allocation of all data atoms and also getting very close to a Nash equilibrium.  Denote by $W^{\infty}$ the final allocation state of the system after time $T$ has elapsed. 

\item
The parameter $\gamma$ appearing in the Gibbs distribution have been chosen to be time-varying with $$\gamma(t)=\gamma(t-1)+\frac{1}{\lambda_{max}*100}$$ where $\lambda_{max}$ is the maximum reliability of the resources. This is a typical choice done on such best response dynamics, even if the theoretical result expressed below can not directly be applied to insure convergence. In real world applications, where adaptation to a time-varying scenario (e.g. addition or deletion of units, change in the topology) is needed, $\gamma$ must be instead kept bounded. 
\item We have assumed units to have the same free space to offer $\beta = 50$ and to have possibly different amounts $\alpha_x$ to allocate. 
\item We have assumed units to split into two subsets of equal size $\mc X_1$ and $\mc X_2$ characterized by two different reliability levels, respectively,  $\lambda_1=0.5$ and $\lambda_2=0.8$. 
\item The congestion parameter is chosen to be $k_c=1$ while we will consider different values for $k_a$.
\end{itemize}

Moving data from one resource to another one is an expensive task which must be carefully monitored in real applications. To this aim, we have introduced the index $\nu_{moves}$ which computes the number of allocation or distribution moves per piece of data throughout the dynamics. In formula, if $m_i$ is the total number of moves performed by agent $i$ during the run of the algorithm, we put
$$\nu_{moves}=\frac{1}{n}\sum\limits_{x\in\mc X}\frac{m_i}{\alpha_i}.$$

The global utility of the system in the allocation state $W$ is defined as
$F(W)=\sum_{x,y\in\mc X} W_{xy}f_{xy}(W)$. 
Put
$F^*=\max_{W\in\mc W}F(W)$. 
$$\rho :=\frac{F(W^{\infty})}{F^*}$$
measures how the solution found is performing with regards to the global utility (notice that the potential $\Psi$ does not coincide with $F$ so that $W^{\infty}$  is not a-priori a maximum of $F$).

Units give preference to the most trusted resource in $\mc X_2$, however, depending on the amount of data they need to allocate and on the type of resources in their neighborhood, they need to use also less trusted resources in $\mc X_1$. We define the mean and the variance of the satisfaction level as
$$\bar\Lambda:=\frac{1}{n}\sum_{x\in\mc X}\sum_{y\in N_x}\frac{W^{\infty}_{xy}}{\alpha_x}\lambda_y,$$
$$\Lambda_{var}:=\frac{1}{n}\sum_{x\in\mc X}\left(\sum_{y\in\mc X}\frac{W^{\infty}_{xy}}{\alpha_x}\lambda_y-\bar\Lambda\right)^2$$
If $\lambda_1$ and $\lambda_2$ are taken to be the probability that if contacted at a random time the resource is available to give access to the stored data, $\bar\Lambda$ can be interpreted as the probability that a piece of data can be recovered when requested at some random time. 
%Of course $\bar\Lambda$ should be confronted with the level of satisfaction obtained if all the most trustable resources have been used, namely
%$$\Lambda_{max}:=

The presence of the congestion term in the utility function should insure that all resources with the same $\lambda$ should in principle be used equally. We measure the mean and variance of the congestion level of resources in $\mc X_i$ by
$$\bar C^i:=\frac{1}{n\beta }\sum_{y\in\mc X_i}W^{\infty}_y,$$
$$C_{var}^i:=\frac{1}{n }\sum_{y\in\mc X_i}\left(\beta^{-1}W^{\infty}_y-\bar C\right)^2$$

Finally, we consider the in and out mean degrees measuring the topological complexity of the subgraph consisting of the edges $(x,y)$ for which $W_{xy}^{\infty}>0$:
$$d^+:=\frac{1}{n}\sum_{x\in\mc X}\sum_{y\in\mc X}\1_{\{W_{xy}>0\}},$$
$$d^-_i:=\frac{2}{n}\sum_{x\in\mc X}\sum_{y\in\mc X_i}\1_{\{W_{xy}>0\}},\; i=1,2$$

Below we present four sets of examples. The first two are with $n=50$ units, $\alpha_x$ constant, and different topologies. The third one deals instead with a situation where units have different $\alpha_x$. Finally, the fourth example deals with a larger number of users.
%In this section we show some examples that explain the behavior of the game. For all the example we consider the underlying graph to be static i.e. it does not change during time. The first example is to compare the game changing the underlying network: we show that if the underlying graph is complete there are many unused edges while considering a regular graph with a lower degree the allocation problem is still solvable. the last two example instead show the asymptotic behavior of the game.

%To understand the game we compute some interesting parameter: $e$ is the fraction of edges used by the user with respect to the total of available edges in the underlying graph (notice that the graph is direct which means that an undirected edge counts as two edges); the quantity $\nu_{moves}$ is defined as follows
%$$\nu_{moves}=\frac{1}{n}\sum\limits_{i=1}^n\frac{\text{real moves of agent $i$}}{\alpha_i/\Delta}$$
%and indicates how many times each agent effectively make a move (allocation or distribution); $\bar Q$ and $Q_{var}$ represent the average level of satisfaction of the users and its variance; finally, $\bar C$ and $C_{var}$ indicates the average level of congestion of the resources and its variance. In each example the congestion parameter is fixed, $k_c=1$.

\begin{example}
Consider a case with $n=50$ users and $\alpha_x=\alpha=45$ for every unit $x$. It follows from Proposition \ref{prop allocation} that allocation is possible in any regular graph. Also notice that Theorem \ref{theo main 1} can be applied.
We first assume that the graph is complete and then we analyze the case of a regular graph of degree $10$.

The following table shows the value of the various performance parameters varying the aggregation term $k_a$.

\begin{center}
\begin{table}[h]
\caption{Complete Graph}
\centering
\begin{tabular}{|l|c|c|c|}
\hline
&$k_a=0$ & $k_a=0.25$ & $k_a=0.45$ \\
\hline
$\nu_{moves}$&   1.6271   &  1.3068    &    1.2548      \\
$\bar \Lambda$ &  0.6667   &   0.6592   &    0.6593  \\
$\Lambda_{var}$ &    $6.4818*10^{-4}$  &  0.0119    &   0.0122   \\
$\bar C^1$ &   0.8000  &   0.8450  &    0.8442  \\
$C_{var}^1$ &    $9.5680*10^{-4}$  &   0.1149   &   0.1195  \\
$\bar C^2$ &  1   &  0.9550    &   0.9558   \\
$C_{var}^2$ &  0    &   0.0280   &  0.0261   \\
$d^+$ &   44.8460    &    9.5420   &  9.6720   \\
$d^-_1$ &  43.9280    &   9.1720   & 9.1280    \\
$d^-_2$ &   45.7640   &   9.9120   &  10.2160   \\
%c_var 0.0102  & 0.0301 & 0.0313
%bar_c=0.9000 &  0.9000 & 0.9000
$\rho$ & 0.9787 & 0.6812 & 0.6796 \\
\hline
    
\end{tabular}
\end{table}
\end{center}

The fact that, on average, each user makes less than two moves per each atom, suggests that at the end of the allocation the system is already near to the equilibrium. 

Notice the effect of the aggregation term on this parameter: when $k_a=0$, the parameter $\nu_{moves}$ is considerably higher as agents tend to move their data to leverage the congestion level over all the available resources. When the aggregation term is present this phenomenon is much reduced as agents have an incentive to keep data atoms together. The congestion indices confirm that agents preferably allocate over more trusted resources and that resources are used sufficiently equally. We can also see that the aggregation term brings to a slightly greater usage of worse resources.
%Another interesting fact is that even if the average level of congestion is similar for every case, it is clear that users prefer to use more reliable resources. We can also see that the aggregation term brings to a greater usage of worse resources but the fact that there is an higher variance implies that most of the devices with lower reliability are still less used. 
The parameters where the influence of the aggregation term is even clearer, are the degrees: in the first case almost every edge is used while in the second two cases the number is sensibly lower.
%; moreover, the in-degree parameters confirm that more reliable resources are more used. 
Finally, concerning the utility parameter ratio $\rho$, notice that our algorithm performs very well for $k_a=0$ while we witness a certain degradation when $k_a>0$.  
%a nearly optimal solution; in the other cases our solution is further from the optimal but we must consider that each user can do very few moves and that our algorithm is discrete.

%As in the tables, the number of used edges is considerably lower in the second case and, while in the first case all resources are used in a similar way, increasing the aggregation factor some resources are even full.
\end{example}

\begin{example}
Consider now the same case with 50 users where the underlying graph is regular with degree $10$. As in the previous case, the table shows the parameters in the case with $k_a=0$, $k_a=0.25$ and $k_a=0.45$.

\begin{center}
\begin{table}[h]{\label{deg10_45}}
\caption{Regular Degree 10}
\centering
\begin{tabular}{|l|c|c|c|}
\hline
&$k_a=0$ & $k_a=0.25$ & $k_a=0.45$ \\
\hline
$\nu_{moves}$& 1.4187     &    1.2185   &    1.1714      \\
$\bar \Lambda$ & 0.6667  & 0.6596     &      0.6606    \\
$\Lambda_{var}$ &  0.0019    &  0.0136    &      0.0143     \\
$\bar C^1$ &  0.8000  &  0.8422    &    0.8364      \\
$C_{var}^1$ &  0.0011   &  0.1214    &   0.1350       \\
$\bar C^2$ & 1  &   0.9578   &    0.9636      \\
$C_{var}^2$ &  0  &   0.0261   &     0.0251     \\
$d^+$ & 9.9560   &  6.2580    &   6.3700       \\
$d^-_1$ & 9.9240 &  5.9400    &    6.2520      \\
$d^-_2$ &  9.9880  & 6.5760     &      6.4880    \\
%c_var & 0.0102  & 0.0280 & 0.9000
%bar_c & 0.9000  & 0.9000 & 0.0339
$\rho$ & 0.9784 & 0.8872 & 0.9297\\
\hline 
\end{tabular}
\end{table}
\end{center}
These data confirm the good performance properties of the algorithm even when the graph of connection is of considerably less complexity. Notice how, when the aggregation term is present, the number of resources used per agent is something above $6$ to be compared to the previous case where it was almost $10$. In other terms, imposing a constrained communication pattern, it does not degrade performance and helps keeping complexity at a lower level.

%Even in this asset we have different values for $\nu_{moves}$ varying $k_a$ (for the same reason of the previous example) and the same influence of the aggregation term in the choice of resources. What is interesting is the fact that even if the number of edges is lower than in the complete case, the allocation is still possible and adding the aggregation parameter we can see that even less edges are enough. Regarding the utility parameter, as is in the previous case we have that in the case with $k_a=0$ our payoff is the nearest to the optimal while in the second two cases the utility is lower.

Figures $1$ and $2$ show the difference between the graph we impose and the graph with edges that are used by the users in the case with $k_a=0.25$.

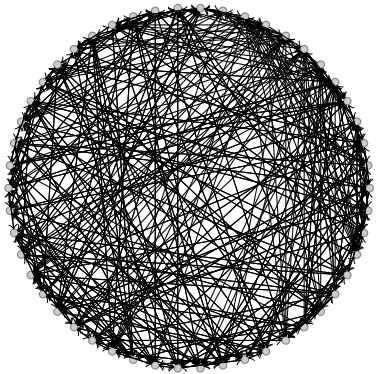
\begin{figure}[h]\label{grafo_un}
\centering
\begin{tikzpicture}
  [scale=2.4,auto=left,every node/.style={circle,draw=black!50,scale=.3,fill=black!20}]
  
  \node (n1) at (0.9921,0.1253){};  
  \node (n2) at (0.9686,0.2487){}; 
  \node (n3) at (0.9298,0.3681){}; 
  \node (n4) at (0.8763,0.4818){}; 
  \node (n5) at (0.8090,0.5878){}; \node (n6) at (0.7290,0.6845){}; \node (n7) at (0.6374,0.7705){}; \node (n8) at (0.5358,0.8443){}; \node (n9) at (0.4258,0.9048){}; \node (n10) at (0.3090,0.9511){}; \node (n11) at (0.1874,0.9823){}; \node (n12) at (0.0628,0.9980){}; \node (n13) at (-0.0628,0.9980){}; \node (n14) at (-0.1874,0.9823){}; \node (n15) at (-0.3090,0.9511){}; \node (n16) at (-0.4258,0.9048){}; \node (n17) at (-0.5358,0.8443){}; \node (n18) at (-0.6374,0.7705){}; \node (n19) at (-0.7290,0.6845){}; \node (n20) at (-0.8090,0.5878){}; \node (n21) at (-0.8763,0.4818){}; \node (n22) at (-0.9298,0.3681){}; \node (n23) at (-0.9686,0.2487){}; \node (n24) at (-0.9921,0.1253){}; \node (n25) at (-1.0000,-0.0000){}; \node (n26) at (-0.9921,-0.1253){}; \node (n27) at (-0.9686,-0.2487){}; \node (n28) at (-0.9298,-0.3681){}; \node (n29) at (-0.8763,-0.4818){}; \node (n30) at (-0.8090,-0.5878){}; \node (n31) at (-0.7290,-0.6845){}; \node (n32) at (-0.6374,-0.7705){}; \node (n33) at (-0.5358,-0.8443){}; \node (n34) at (-0.4258,-0.9048){}; \node (n35) at (-0.3090,-0.9511){}; \node (n36) at (-0.1874,-0.9823){}; \node (n37) at (-0.0628,-0.9980){}; \node (n38) at (0.0628,-0.9980){}; \node (n39) at (0.1874,-0.9823){}; \node (n40) at (0.3090,-0.9511){}; \node (n41) at (0.4258,-0.9048){}; \node (n42) at (0.5358,-0.8443){}; \node (n43) at (0.6374,-0.7705){}; \node (n44) at (0.7290,-0.6845){}; \node (n45) at (0.8090,-0.5878){}; \node (n46) at (0.8763,-0.4818){}; \node (n47) at (0.9298,-0.3681){}; \node (n48) at (0.9686,-0.2487){}; \node (n49) at (0.9921,-0.1253){}; \node (n50) at (1.0000,0.0000){};

 \foreach \from/\to in
{n1/n4,n1/n11,n1/n15,n1/n16,n1/n18,n1/n23,n1/n32,n1/n33,n1/n40,n1/n46,
n2/n3,n2/n7,n2/n19,n2/n21,n2/n26,n2/n29,n2/n35,n2/n38,n2/n40,n2/n50,
n3/n2,n3/n8,n3/n26,n3/n31,n3/n33,   n3/n38,   n3/n45,   n3/n47,   n3/n48,   n3/n49,
n4/n1,   n4/n5,   n4/n8, n4/n11, n4/n23, n4/n24, n4/n25, n4/n27, n4/n33,n4/n41,
     n5/n4, n5/n14,  n5/n16,  n5/n19,  n5/n24,  n5/n25,  n5/n39,  n5/n43,  n5/n48,  n5/n50,
    n5/n17, n5/n24,  n5/n28,  n5/n32,  n5/n37,  n5/n38,  n5/n41,  n5/n43,  n5/n45,  n5/n47,
     n6/n2,n6/n10,  n6/n16,   n6/n23,    n6/n28,     n6/n29,     n6/n43,     n6/n44,     n6/n45,     n6/n47,
      n7/n3,     n7/n4,n7/n13,    n7/n30,    n7/n31,    n7/n32,    n7/n34,    n7/n39, n7/n42,    n7/n47,
    n8/n11,   n8/n12,n8/n16,    n8/n18,    n8/n20,    n8/n22,    n8/n27,    n8/n29, n8/n30,    n8/n34,
     n9/n7, n9/n12,n9/n25,  n9/n36,  n9/n42,  n9/n44,  n9/n46,  n9/n47,    n9/n48, n9/n49,
  n10/n1,   n10/n4,   n10/n9,    n10/n12,   n10/n20,    n10/n22, n10/n28,    n10/n29,    n10/n34,   n10/n50,
  n12/n9, n12/n10,  n12/n11,n12/n15,    n12/n17,    n12/n22,    n12/n26,    n12/n28,    n12/n31,    n12/n33,
 n13/n8,n13/n18, n13/n20,   n13/n31,   n13/n32,   n13/n37,   n13/n42,   n13/n44,   n13/n46,   n13/n47,
n14/n5,    n14/n16,n14/n18,  n14/n19,  n14/n25,  n14/n26,  n14/n28,  n14/n36,  n14/n41,  n14/n48,
 n15/n1,   n15/n12,    n15/n16, n15/n21, n15/n22, n15/n31, n15/n34, n15/n38, n15/n39, n15/n41,
 n16/n1,   n16/n5,   n16/n7,     n16/n9,    n16/n14,    n16/n15, n16/n18,   n16/n35,    n16/n39,    n16/n49, n17/n6, n17/n12,  n17/n18,    n17/n24,    n17/n31,    n17/n34,   n17/n35, n17/n36,    n17/n37,   n17/n48,     n18/n1, n18/n9,n18/n13,  n18/n14,  n18/n16,  n18/n17, n18/n25,  n18/n35,    n18/n39, n18/n42,
   n19/n2,      n19/n5, n19/n14,     n19/n20,  n19/n22,  n19/n30,  n19/n37,    n19/n38,   n19/n41,    n19/n44,
      n20/n9,    n20/n11,    n20/n13,    n20/n19,n20/n25,n20/n26,   n20/n44,   n20/n46,    n20/n48,   n20/n49,
     n21/n2,  n21/n15,  n21/n23,  n21/n28,   n21/n37,  n21/n38,   n21/n40,   n21/n45,  n21/n47,   n21/n50,
     n22/n9, n22/n11, n22/n12, n22/n15,n22/n19, n22/n24,    n22/n36,    n22/n41, n22/n44,    n22/n49,
     n23/n1, n23/n4,n23/n7,    n23/n21,   n23/n30,    n23/n31,   n23/n39,   n23/n40,    n23/n46,   n23/n50,
     n24/n4,n24/n5,n24/n6,   n24/n17,    n24/n22,   n24/n30,    n24/n33,    n24/n36,   n24/n39,    n24/n4,
     n25/n4,n25/n5,    n25/n10,   n25/n14,n25/n18,   n25/n20, n25/n26, n25/n27,   n25/n35, n25/n40,
     n26/n2,n26/n3,   n26/n12,  n26/n14,n26/n20,  n26/n25,  n26/n35,  n26/n36,  n26/n38,  n26/n45,
     n27/n4,n27/n9,  n27/n25, n27/n30,n27/n32, n27/n35,   n27/n38,   n27/n42, n27/n44,   n27/n46,
     n28/n6,n28/n7, n28/n11,n28/n12, n28/n14,n28/n21,    n28/n35,    n28/n40,n28/n45,    n28/n46,
     n29/n2,n29/n7,   n29/n9, n29/n11,     n29/n32,   n29/n34,n29/n36, n29/n43,  n29/n45, n29/n49,
     n30/n8,n30/n9, n30/n19, n30/n23,  n30/n24,     n30/n27,   n30/n35,     n30/n40,     n30/n41,     n30/n43,
      n31/n3,    n31/n8,n31/n12,  n31/n13,    n31/n15,    n31/n17,   n31/n23,n31/n41, n31/n48, n31/n50,
     n32/n1,  n32/n6,n32/n8,    n32/n13,   n32/n27, n32/n29,  n32/n33,    n32/n39,n32/n49,n32/n50,
     n33/n1, n33/n3,n33/n4,   n33/n12,    n33/n24,n33/n32,   n33/n34,   n33/n37,n33/n42,n33/n45,
     n34/n8,n34/n9,n34/n11,   n34/n15,n34/n17,n34/n29,n34/n33,   n34/n36, n34/n39, n34/n43,
     n35/n2,    n35/n16,  n35/n17,   n35/n18,  n35/n25,n35/n26,  n35/n27,   n35/n28,   n35/n30,   n35/n44,
    n36/n10,    n36/n14,    n36/n17,   n36/n22,    n36/n24,n36/n26,   n36/n29,  n36/n34, n36/n37,    n36/n41,
     n37/n6,  n37/n13,    n37/n17, n37/n19,    n37/n21,    n37/n33,    n37/n36,n37/n40, n37/n46,    n37/n48,
n38/n2, n38/n3,n38/n6,n38/n15,n38/n19, n38/n21,n38/n26, n38/n27,n38/n43,n38/n44,
     n39/n5,n39/n8, n39/n15,  n39/n16,n39/n18,    n39/n23,n39/n24,   n39/n32,    n39/n34,    n39/n50,
     n40/n1,n40/n2,n40/n21,n40/n23,n40/n25,n40/n28,n40/n30,n40/n37,n40/n41,n40/n50,
     n41/n4,     n41/n6,    n41/n14,   n41/n15,    n41/n19,    n41/n22,n41/n30,n41/n31,n41/n36,n41/n40,
     n42/n8,n42/n10,n42/n13,n42/n18,n42/n27,n42/n33,n42/n45,n42/n46,n42/n47,n42/n48,
     n43/n5,n43/n6,n43/n7,n43/n24,n43/n29,n43/n30,n43/n34,n43/n38,n43/n45,n43/n49, 
     n44/n7,n44/n10,n44/n13,n44/n19,n44/n20,n44/n22,n44/n27,n44/n35,n44/n38,n44/n49,
     n45/n3, n45/n6, n45/n7, n45/n21, n45/n26, n45/n28, n45/n29, n45/n33, n45/n42, n45/n43, 
     n46/n1, n46/n10,n46/n13,n46/n20,n46/n23,n46/n27,n46/n28,n46/n37,n46/n42,n46/n49,
     n47/n3, n47/n6, n47/n7, n47/n8, n47/n10, n47/n13, n47/n21, n47/n42, n47/n48, n47/n50, 
     n48/n3, n48/n5, n48/n10, n48/n14, n48/n17, n48/n20, n48/n31, n48/n37, n48/n42, n48/n47, 
     n49/n3, n49/n10, n49/n16, n49/n20, n49/n22, n49/n29, n49/n32, n49/n43, n49/n44, n49/n46, 
     n50/n2, n50/n5, n50/n11, n50/n21, n50/n23, n50/n31, n50/n32, n50/n39, n50/n40, n50/n47}
     \draw [-to] (\from) -- (\to);

\end{tikzpicture}   
\caption{Underlying Network}
\end{figure} 

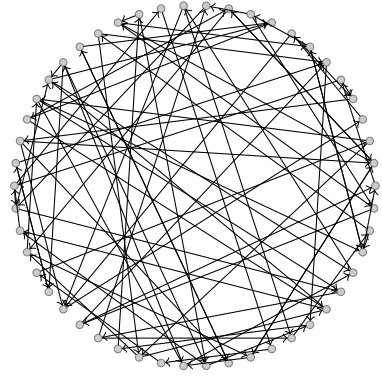
\begin{figure}[h]\label{grafo_ue}
\centering
\begin{tikzpicture}
  [scale=2.4,auto=left,every node/.style={circle,draw=black!50,scale=.3,fill=black!20}]
\node[fill=none,draw=none] (temp) at (0,1.05){};
  
  \node (n1) at (0.9921,0.1253){};  
  \node (n2) at (0.9686,0.2487){}; 
  \node (n3) at (0.9298,0.3681){}; 
  \node (n4) at (0.8763,0.4818){}; 
  \node (n5) at (0.8090,0.5878){}; \node (n6) at (0.7290,0.6845){}; \node (n7) at (0.6374,0.7705){}; \node (n8) at (0.5358,0.8443){}; \node (n9) at (0.4258,0.9048){}; \node (n10) at (0.3090,0.9511){}; \node (n11) at (0.1874,0.9823){}; \node (n12) at (0.0628,0.9980){}; \node (n13) at (-0.0628,0.9980){}; \node (n14) at (-0.1874,0.9823){}; \node (n15) at (-0.3090,0.9511){}; \node (n16) at (-0.4258,0.9048){}; \node (n17) at (-0.5358,0.8443){}; \node (n18) at (-0.6374,0.7705){}; \node (n19) at (-0.7290,0.6845){}; \node (n20) at (-0.8090,0.5878){}; \node (n21) at (-0.8763,0.4818){}; \node (n22) at (-0.9298,0.3681){}; \node (n23) at (-0.9686,0.2487){}; \node (n24) at (-0.9921,0.1253){}; \node (n25) at (-1.0000,-0.0000){}; \node (n26) at (-0.9921,-0.1253){}; \node (n27) at (-0.9686,-0.2487){}; \node (n28) at (-0.9298,-0.3681){}; \node (n29) at (-0.8763,-0.4818){}; \node (n30) at (-0.8090,-0.5878){}; \node (n31) at (-0.7290,-0.6845){}; \node (n32) at (-0.6374,-0.7705){}; \node (n33) at (-0.5358,-0.8443){}; \node (n34) at (-0.4258,-0.9048){}; \node (n35) at (-0.3090,-0.9511){}; \node (n36) at (-0.1874,-0.9823){}; \node (n37) at (-0.0628,-0.9980){}; \node (n38) at (0.0628,-0.9980){}; \node (n39) at (0.1874,-0.9823){}; \node (n40) at (0.3090,-0.9511){}; \node (n41) at (0.4258,-0.9048){}; \node (n42) at (0.5358,-0.8443){}; \node (n43) at (0.6374,-0.7705){}; \node (n44) at (0.7290,-0.6845){}; \node (n45) at (0.8090,-0.5878){}; \node (n46) at (0.8763,-0.4818){}; \node (n47) at (0.9298,-0.3681){}; \node (n48) at (0.9686,-0.2487){}; \node (n49) at (0.9921,-0.1253){}; \node (n50) at (1.0000,0.0000){};

 \foreach \from/\to in
{n1/n11,    n1/n16,    n1/n40, n2/n21,    n2/n35, n3/n8, n4/n23, n5/n4,    n5/n25, n6/n43,    n7/n10,    n7/n47,n8/n4,    n8/n31,    n9/n22,    n10/n12,    n10/n49, n11/n1,    n11/n28,    n12/n26,    n13/n31,    n13/n42,    n14/n41,    n15/n16,    n15/n38,     n16/n7,    n16/n49,     n17/n6,    n17/n35,     n18/n9,    n18/n39,n19/n20,    n19/n30,    n19/n38,     n20/n9,    n20/n46,    n21/n37,    n21/n45,    n21/n47,    n22/n11,    n22/n44, n23/n1,    n23/n31,    n23/n39,     n24/n6,    n24/n30,    n25/n14,    n25/n26,    n26/n20,    n26/n45, n27/n25,n28/n21,    n28/n40,     n29/n2,    n30/n24,     n31/n3,     n32/n6,    n32/n50,    n33/n12,    n33/n42,    n34/n15,n35/n28,    n35/n44,    n36/n29,    n37/n19,    n37/n46,    n38/n27,    n38/n43,    n39/n18,     n40/n1,    n40/n37,    n41/n36,    n42/n48,    n43/n34,    n44/n20,    n44/n27,    n45/n33,    n46/n13,     n47/n7,    n47/n50,    n48/n17,    n48/n47,    n49/n32,     n50/n5}
    \draw[-to] (\from) -- (\to);

\end{tikzpicture}   
\caption{Used Edges}
\end{figure}
\end{example}
Next example shows that variability of the amount of the data to be allocated by the various units does not significantly change the overall performance of the algorithm.
\begin{example}
In this example we consider the same number of $n=50$ users split into five equal subfamilies with varying $\alpha= 35, 40, 45, 50, 55$ (so that the average is still $45$). The underlying graph is assumed to be random regular with degree $10$ and $k_a=0.45$. A-priori there is no guarantee that an allocation exists in this case, however, the algorithm always finds one. The table below reports the value of the various performance indices which turn out to be quite close to the case of constant $\alpha_x$.

\begin{center}
\begin{table}[h]
\caption{Different $\alpha$}
\centering
\begin{tabular}{|l|c|}
\hline

\hline
$\nu_{moves}$ &      1.1552  \\
$\bar \Lambda$ &   0.6613 \\
$\Lambda_{var}$ &    0.0138 \\
$\bar C^1$ &    0.8387   \\
$C_{var}^1$ &    0.1464    \\
$\bar C^2$ & 0.9613   \\
$C_{var}^2$ &    0.0328    \\
$d^+$ &    6.4040   \\
$d^-_1$ &    6.1200   \\
$d^-_2$ &    6.6880   \\
%c_var & 0.0377 
%bar_c & 0.9000  
\hline
    
\end{tabular}
\end{table}
\end{center}
\end{example}
%Comparing this table with table II we can see that at the end of the game, the state of the system is not very different from the case with all the users having the same $\alpha$. In fact, the number of moves is similar to the previous case and users still prefer to use more reliable resources, as shown by the congestion and in-degree parameters.

Finally, next example shows that the algorithm has good scalability properties.

\begin{example}
In this example we consider and compare three communities with $n=50, 100, 1000$ always connected through a random regular graph of degree $10$. We assume that $\alpha_x=45$ for all units. The aggregation parameter is $k_a=0.45$.

\begin{center}
\begin{table}[h]
\caption{Asymptotic Behavior}
\centering
\begin{tabular}{|l|c|c|c|}
\hline
&50 & 100 & 1000 \\
\hline
$\nu_{moves}$&    1.1714     &  1.1490  & 1.1304 \\
$\bar \Lambda$ &       0.6606 &  0.6605  &  0.6566 \\
$\Lambda_{var}$ &        0.0143   &  0.0146  & 0.0114 \\
$\bar C^1$     &    0.8364  &  0.8370  &  0.8604  \\
$C_{var}^1$ &   0.1350   & 0.1262   &     0.1068\\
$\bar C^2$ &    0.9636   &   0.9630  &  0.9396 \\
$C_{var}^2$ &      0.0251 &   0.0183  &   0.0616 \\
$d^+$ &   6.3700    &  6.2840  &  6.1902 \\
$d^-_1$ &    6.2520   &   5.9380  &  6.0004 \\
$d^-_2$ &        6.4880  &  6.6300   &  6.3800 \\
%c_var & 0.0339  & 0.0307 & 0.0323
%bar_c & 0.9000  & 0.9000 & 0.9000
\hline
    
\end{tabular}
\end{table}
\end{center}
\end{example}
%It is intuitive that there are no significant differences between the congestion and satisfaction levels in the three cases. The interesting fact is that also the number of moves is similar which means that the algorithm scales linearly in the number of agents.

\section{Conclusions}
In this paper we have proposed, in a game theoretic framework, a peer-to-peer decentralized storage model where a network of units are, at the same time, end users needing to allocate externally a back up of their data, as well storage resources for other users. We have proposed a novel fully distributed algorithm where units, connected through a network, activate autonomously at random time and either allocate or move pieces of their data among the neighboring available resources. Actions taken by the units are noisy best response actions with respect to utility functions which incorporate the congestion of the resources, their reliability, as well possibly, the fragmentation of the stored data. The algorithm has been claimed to converge, with probability $1$, to a Nash equilibrium of the game and several numerical simulations have here validated this claim. In a forthcoming paper, we will propose a detailed theoretical analysis of our algorithm.

We believe that there are several challenges related to the peer-to-peer storage model which have not yet been satisfactorily addressed by pure mathematical model. Some of them are reported below and will be the goals of our future research.

\begin{itemize}
\item In realistic scenarios, for security reasons, units need to allocate more than one copy of their own data. This poses new issues and constraints as it becomes fundamental that copies of the same data are not stored in the same resource. Most of our theoretical results, including the matching problem equivalence, will need to be completely revisited to be applied in this new scenario.

\item Our simulations show that the algorithm has very good convergence properties. It will be of interest to establish theoretical bounds on the convergence time.

\item In our model, units, thought as resources, are assumed to be always on and available, if not yet fully congested, for storage actions. It would be of interest to let the possibility that resources might be off at certain times and to connect this behavior to the reliability parameter present in the utility functions.

\item Units, in our model, are completely anonymous and resources do not make any filter on new allocation requests. More interesting models should incorporate trust formation mechanisms where more trusted units (when thought as resources) should have a vantage in finding place to store their data.
\end{itemize}

\section*{Acknowledgment}
We acknowledge that this work has been done while Barbara Franci was a PhD Student sponsored by a Telecom Italia grant.


\begin{thebibliography}{99}

\bibitem{CRAG} V. Jalaparti, G. Nguyen, I.Gupta, M. Caesar, "Cloud Resource Allocation Games" ,Technical Report, University of Illinois. http://hdl.handle.net/2142/17427
\bibitem{ROS} R. W. Rosenthal, "A Class of game Possessing Pure-Strategy  Nash Equilibrium",Int. J. GameTheory 2, 65Ð67.
\bibitem{HB} T. Roughgarden, E. Tardos, "How Bad is Selfish Routing?",Journal of  ACM,vol. 49, no. 2, pp. 236Ð259, 2002.
\bibitem{ACG} C. Tekin, M. Liu, R. Southwell, J. Huang, S.H. A. Ahmad, "Atomic Congestion Game on Graphs and its Application in Networking",IEEE/ACM Transaction on Networking , vol. 20 , no. 5 , pp.1541 -1552 , 2012 
\bibitem{TEM} H. Tembine, E. Altman, R. El-Azouzi, Y. Hayel, "Evolutionary Games in Wireless Networks",IEEE Trans. on Systems, Man, and Cybernetics, Part B, Cybernetics, vol. 40, pp. 634-646, 2010. 
\bibitem{SAN} W.H. Sandholm, "Population Games and Evolutionary Dynamics" Cambridge, MA: The MIT press, 2010.

\end{thebibliography}
\end{document}